\documentclass[
 letterpaper,
 onecolumn,
 aps, accepted=2023-07-20
]{quantumarticle}
\pdfoutput=1
\usepackage[utf8]{inputenc}
\usepackage[english]{babel}
\usepackage{hyperref}
\usepackage{braket}
\usepackage{amsmath}
\usepackage{amssymb}
\usepackage{amsthm}
\usepackage{todonotes}
\usepackage{subcaption}
\usepackage{ulem}
\usepackage[sort&compress,numbers]{natbib}

\DeclareMathOperator\erf{erf}

\DeclareMathOperator{\poly}{poly}
\DeclareMathOperator{\polylog}{polylog}
\DeclareMathOperator{\Real}{Re}

\begin{document}
\title{Nearly-frustration-free ground state preparation}
\author{Matthew Thibodeau}
\email{mt24@illinois.edu}
\author{Bryan K. Clark}
\email{bkclark@illinois.edu}
\affiliation{Department of Physics, University of Illinois at Urbana-Champaign, IL 61801, USA}
\affiliation{IQUIST and Institute for Condensed Matter Theory and NCSA Center for Artificial Intelligence Innovation, University of Illinois at Urbana-Champaign, IL 61801, USA}

\maketitle

\newtheorem{defn}{Definition}
\newtheorem{thm}{Theorem}
\newtheorem{lem}{Lemma}
\newtheorem{prop}{Proposition}
\newtheorem{cor}{Corollary}

\begin{abstract}
    Solving for quantum ground states is important for understanding the properties of quantum many-body systems, and quantum computers are potentially well-suited for solving for quantum ground states.  Recent work \cite{Lin2020} has presented a nearly optimal scheme that prepares ground states on a quantum computer for completely generic Hamiltonians, whose query complexity scales as $\delta^{-1}$, i.e. inversely with their normalized gap.  Here we consider instead the ground state preparation problem restricted to a special subset of Hamiltonians, which includes those which we term ``nearly-frustration-free'': the class of Hamiltonians for which the ground state energy of their block-encoded and hence normalized Hamiltonian $\alpha^{-1}H$ is within $\delta^y$ of -1, where $\delta$ is the spectral gap of $\alpha^{-1}H$ and $0 \leq y \leq 1$. For this subclass, we describe an algorithm whose dependence on the gap is asymptotically better, scaling as $\delta^{y/2-1}$, and show that this new dependence is optimal up to factors of $\log \delta$.  In addition, we give examples of physically motivated Hamiltonians which live in this subclass. Finally, we describe an extension of this method which allows the preparation of excited states both for generic Hamiltonians as well as, at a similar speedup as the ground state case, for those which are nearly frustration-free.
\end{abstract}

\section{Introduction}

Steady progress in the development of quantum computing hardware has spurred interest in possible applications of these machines. One natural candidate is the ground state preparation problem: given a many-body Hamiltonian $H$, prepare its ground state in a quantum register. Given a quantum many-body system, physicists are often interested in the properties of its ground state because this state predominates at low temperature; preparing ground states using a quantum computer is an important first step to studying them synthetically, with full control over their generating Hamiltonian.

To prepare the ground state of $H$ in a quantum register, one must distinguish it from all other states in the spectrum. Since, by definition, the ground state is separated in energy from the rest of the spectrum by the spectral gap $\Delta$, it should be no great surprise that $\Delta$ is a key parameter controlling the complexity of quantum algorithms that prepare ground states. For example, both phase estimation \cite{Kitaev1995} combined with amplitude amplification \cite{Brassard2002}, as well as quantum singular value tranformation (QSVT) \cite{Gilyen2019}, can find the ground state in time $O(1/\Delta)$.  The latter is more efficient with respect to the error $\epsilon$, with a cost which goes as $\log(1/\epsilon)$ instead of $1/\epsilon$ \cite{Ge2019}.  In fact, the QSVT is known to be optimal with respect to query complexity for preparing the ground state in a quantum register for generic Hamiltonians \cite{Lin2020}. The query complexity  measures the number of times that $H$ must be accessed and roughly corresponds to the circuit depth. The optimality of this procedure comes from a reduction to Grover's algorithm \cite{Grover96}.

For QSVT, beyond $\Delta$ and $\epsilon$, there is an additional dependence on the overlap $\gamma = |\braket{\phi|\mu}|$ with the starting ansatz, and an additional term $\alpha$ related to the cost of loading (i.e. block-encoding) $H$ into a quantum computer; roughly, for generic local lattice models, a block-encoding with $\alpha \propto N_\textrm{sites}$ is easily obtained  \cite{Low2019}.
In fact, both in our algorithms and those previously given, $\Delta$ always appears in the complexity normalized to $\Delta/\alpha\sim \Delta/N$;
Section 3 covers the details of this normalization.

In many situations, one might know some special properties that $H$ may satisfy which then may evade this general optimality result.  Examples of potential properties could include $k$-locality, frustration-free ground states, or stoquastic Hamiltonians \cite{Bravyi2008, Somma2013}.   It is useful, then, to find quantum circuits that prepare the ground state of $H$ after assuming that $H$ is a member of some restricted subclass of all Hermitian operators, or to bound the minimum depth of such circuits.

One important subclass is that of frustration-free Hamiltonians, where the ground state of $H$ is simultaneously a ground state of all the operators which sum up to make $H$. A naive application of the algorithm from ref. \cite{Lin2020} still scales as $1/\Delta$ even in the frustration free case.  Nonetheless, there is a straightforward way to do better for many frustration free Hamiltonians using spectral gap amplification \cite{Somma2013} or uniform spectral amplification \cite{Low2017}. These two approaches take Hamiltonians which are frustration free when decomposed into operator terms which are projectors or linear combination of unitaries respectively and generates a new Hamiltonian with the same ground state which has an $O(1)$ gap.  Given a standard (block-encoded) description of our original Hamiltonian, one query of the new Hamiltonian takes $O(1/\sqrt{\Delta})$ queries of the original Hamiltonian.  Using the optimal ground-state generation on the new Hamiltonian then takes $O(1)$ time (given that the gap is $O(1)$) leading to a total number of queries (and circuit depth) which scales as $O(1/\sqrt{\Delta})$, significantly improving over the generic result of $1/\Delta$.  For the case of a local Hamiltonian, this also gives a speedup from $O(N)$ to $O(\sqrt{N})$ via the $\alpha$-dependence of the complexity. The apparent violation of the optimality relation is not surprising, as the optimality theorem of \cite{Lin2020} relies on a specific Hamiltonian $H_G$ that is extremely frustrated.

The improvement of the frustration free case from $1/\Delta$ to $1/\sqrt{\Delta}$ motivates asking whether there is some intermediate class of $H$ between fully frustration-free and the generic case where a speedup might be achieved, as well as a corresponding proof of optimality. 
Further, while this all applies to ground state problems, there are analogous questions one can ask about excited states. 

\section{Summary of Results}

In this work, we address the question of finding a method that prepares a ground state of $H$ for a class of Hamiltonians that are near frustration-free, giving a scaling which continuously tunes between the $1/\Delta$ of the generic case to the $1/\sqrt{\Delta}$ for the frustration-free case. 
While our primary goal will be to show this for ground state problems, we will also extend the work of ref. \cite{Lin2020} to preparing low-lying excited states and furthermore show how they can similarly be accelerated in nearly frustration-free situations. 
To accomplish these goals, the property of nearly frustration-free Hamiltonians that we use is that the ground state energy of the Hamiltonian is close to 
the frustration-free energy.
We then provide a lower-bound on the complexity of preparing ground states of such Hamiltonians,
showing that the scaling with $\Delta$ of our algorithm is optimal up to logarithmic factors.

To be more precise, we need to use the language of \textit{block encoding} \cite{Chakraborty2019, Low2019}.
A block encoding of $H$ is a unitary operator $U$ that has a scaled copy of the Hamiltonian $\alpha^{-1}H$ as its upper-left corner.  
The existence of a block-encoding $U$ of $H$ is a very powerful tool because it allows the application of quantum signal processing \cite{Low2019, Gilyen2019, Haah2019, Low2017PRL, Dong2021}, which lets us transform $H$ by a suitable polynomial $p(x)$, acting on the spectrum of $H$ by $p(H) = \sum_\lambda p(\lambda) \ket{\lambda}\bra{\lambda}$.

Now let us define a parameter $y$ in terms of the other parameters of our problem:
\begin{defn}
Let $H_0$ be a Hermitian operator on $n$ qubits with ground state energy $\mu_0$ and gap $\Delta$, and let $U$ be a block-encoding of $\alpha^{-1}H_0$ (with $\alpha \geq ||H_0||$, the spectral norm of $H_0$) using $m$ extra qubits.  Set $H = \alpha^{-1}H_0$, and let $\delta = \alpha^{-1}\Delta$ and $\mu = \alpha^{-1} \mu_0$; note that $-1 \leq \mu \leq 1$. Given these parameters, there is a $y$ such that $0 \leq y \leq 1$ and $\mu \leq -1 + 2\delta^y$.
\end{defn}
The relationship between $\alpha$, spectral norm $||H_0||$, and the ground state energy $\mu$ will play an important role in the efficiency of our algorithm. In particular, the actual promise we require of $H$ is that there exists a block-encoding where $\alpha$ is close to the magnitude of the ground state energy $|\mu|$. 
This closeness is encoded in the value of the parameter $y$  which specifies the distance between $\mu$ and its minimum possible value of $-1$, and we develop an algorithm for the ground state preparation problem such that the query complexity to $U$ is controlled by $y$. 

\begin{prop}
Let $H_0$ and $U$ be as specified in Definition 1, and let $\ket{\phi}$ satisfy $|\braket{\phi|\mu_0}| = \gamma$. Then the query complexity of preparing the ground state of $H_0$ to precision $\epsilon$ starting with the ansatz $\ket{\phi}$, as measured by queries to $U$ and $U^\dagger$, is 
\begin{align}O\left(\frac{\delta^{y/2 - 1}}{\gamma}\log^{5/2} \frac{1}{\epsilon}\log\frac{1}{\delta}\right)
\end{align}
\end{prop}

To accomplish this, as in \cite{Lin2020}, we will also apply a polynomial of $H$, $p(H)\approx \ket{\mu}\bra{\mu}$, which is approximately the projection onto the ground state $\ket{\mu}$.  Note that the polynomial we apply and that of ref \cite{Lin2020}, though closely related, differs in an important way which is necessary to achieve this improved algorithm. 
In particular, $p(x)$ is necessarily a quickly-oscillating function of $x$ near the ground state energy so as to filter out the first excited state. 
The key improvement for our polynomial is to shift the location of this oscillation from $x = 0$ to outer edges of the domain $[-1,1]$ which allows a reduction in the required degree of $p$ and therefore a reduction in the query complexity. 

As an example, in the case that $H_0$ is frustration-free when decomposed into a linear combination of unitaries, there exists a block-encoding $U$ such that $\alpha = |\mu_0|$, so that $\mu = -1$ and $y = 1$. 
Then, our complexity results agree with that achieved by combining uniform spectral amplification \cite{Low2017} with the standard approach of ref. \cite{Lin2020}, giving that the ground state can be prepared with complexity $O(1/\sqrt{\delta})$.  In addition, when $\mu$ is close enough to $-1$ and $y > 0$, we find an intermediate scaling between $1/\sqrt{\delta}$ and $1/\delta$.

In addition to providing an improved algorithm, we also show that the asymptotic dependence of the query complexity on $\delta$ is optimal, up to a factor of $\log \delta$, for any Hamiltonians which satisfy only the properties of Definition 1.

\begin{prop}
Let $H_0$ be as specified in Definition 1. If $\gamma = \Omega(1)$ and $\delta \to 0+$, then the query complexity of preparing its ground state is $\Omega(\delta^{y/2 - 1})$.
\end{prop}
This is proved through a reduction to Grover's algorithm, using a suitably transformed version of the Hamiltonian $H(1/2)$
in Theorem 10 of \cite{Lin2020}.

\section{Block Encodings, Polynomial Approximation and Degree Reduction}
As stated above, to prepare the ground state we will form a quantum circuit that encodes a projector onto $\ket{\mu_0}$, accurate up to some error $\epsilon$. To do so we will peform quantum signal processing on a block-encoding of $\alpha^{-1}H_0$ for some normalization $\alpha$. We say that $U$ is an $(\alpha, m, \epsilon)$-block-encoding of $H_0$ if
\begin{align}
    || (\bra{0^m}\otimes I) U (\ket{0^m} \otimes I) - \alpha^{-1}H_0 ||_2 \leq \epsilon
\end{align}
i.e. that $U$ is a unitary operator using $m$ extra qubits beyond those used by $H = \alpha^{-1}H_0$, and the restriction of $U$ to the $\ket{0^m}$-flagged subspace is $\epsilon$-close in operator norm to $H$. Given a block-encoding, we can form a polynomial transformations of $H$ using the following theorem:
\begin{thm}[\cite{Gilyen2019}, Theorem 2]
Let U be an $(\alpha, m, 0)$-block-encoding of a Hermitian matrix $H_0$. Let $p \in \mathbf{R}[x]$ be a degree-$l$ even or odd real polynomial with $l > 0$ so that $|p(x)| \leq 1$ for any $x \in [-1,1]$. Then there exists a $(1, m+1, 0)$-block-encoding $\widetilde U$ of $p(\alpha^{-1}H_0)$ using $l$ queries of $U$, $U^\dagger$, and $O((m+1) l)$ other primitive quantum gates.
\end{thm}
The theorem makes clear that the degree $l$ of the polynomial transformation $p$ is the key driver of circuit complexity, and the idea of this work is to show that for some class of $H$, the required $l$ for projection into the ground state can be systematically reduced.

The basic tools of quantum signal processing are these polynomial transformations of the input $H$, most often polynomial approximations of analytic functions. That holds true for our application, where the function at hand is a regulated version of the unit step function which we will take to be $\Theta_k(x) = \tfrac{1}{2}(1 + \erf kx)$ for $k > 0$. 

Note that, if $\Theta(x)$ is the exact step function, then $\Theta(H - x_0 I)$ is exactly the orthogonal projection onto the subspace spanned by eigenstates of $H$ with energy $\leq x_0$. Because we would like to approximate an analytic function, we smooth it with the regulator parameter $k$, which must be taken large enough so that the unit step occurs ``mostly'' (up to some chosen error $\epsilon$) between the ground state and first excited state; this induces a relationship between $k$, spectral gap $\delta$, and $\epsilon$. Because we can only block-encode polynomials of $H$, we will need to accept some approximation error, also bounded by $\epsilon$, in the final projection; for fixed $k$, this error decreases exponentially with the degree $n$ of the polynomial approximation. For a fixed target $\epsilon$, though, the scaling of the degree with the regulator $k$ is actually much worse: generally, we have $n \sim O(k)$, and this worst-case estimate is used extensively in prior work\cite{Lin2020}. However, in certain cases, a strictly lower degree truncation with $n < O(k)$ can achieve the desired accuracy, which we now describe. 

The key to achieving a speedup over the general case is by exploiting a feature of Chebyshev approximations on the interval $x \in [-1, 1]$: quickly-oscillating functions can be approximated (to a given precision) more efficiently when their oscillation occurs near the endpoints $x = \pm 1$. In particular, the Chebyshev approximation of the shifted step function $\Theta_k(x + 1)$ can be truncated at a lower degree than that of $\Theta_k(x)$ while maintaining the same desired precision. We will choose $k$ such that, given an error $\epsilon$, the error function is within $\epsilon$ of the step function outside of a region of width $\delta$:
\begin{align}
    |\Theta_k(x) - \Theta(x)| &\leq \epsilon \quad \quad |x| \geq \delta /2
\end{align}
The details of the approximation are made precise by the following Lemma, which adapts Lemma 16 from \cite{Low2017}.
\begin{lem}
Fix $\epsilon, \delta > 0$ and $\eta$ such that $|\eta| < 1 - \delta$, and let $k = \frac{\sqrt 2}{\delta}\log^{1/2}\frac{1}{2\pi\epsilon^2}$. Then the shifted error function $f(x) = \erf k(x - \eta)$ can be approximated to error $\epsilon$ on the interval $x \in [-1, 1]$ with a polynomial $p(x)$ such that $|p(x)| \leq 1$ for $|x| \leq 1$ and $p$ is of order $n =  O(\frac{\sqrt{1 - |\eta|}}{\delta}\log\frac{\sqrt{1 - |\eta|}}{\delta}\log^{3/2}\frac{1}{\epsilon})= \widetilde O(\frac{\sqrt{1 - |\eta|}}{\delta})$, where the $\widetilde O$ notation hides polylogarithmic factors in $\epsilon, \delta$ and $\eta$.
\end{lem}

The Lemma generalizes the $\eta = 0$ case employed in \cite{Lin2020} and formalizes the asymptotic speedup obtained when $|\eta| > 0$.
\begin{proof}
Expand $f(x) = \sum_j a_j T_j(x)$ in the Chebyshev basis; we will form $p(x)$ by truncating this series. As stated in \cite{Low2017}, the truncation error for $p_n(x) = \sum_{j=0}^n a_j T_j(x)$ is
\begin{align} \label{eq:1}
    \max_{x \in [-1,1]}|f(x) - p(x)| \leq \frac{2M\rho^{-n}}{\rho - 1} = O(\epsilon) \quad \quad M = \max_{z \in E_\rho}|f(x)|
\end{align}
for any $\rho > 1$, where $\rho$ is the parameter of the Bernstein ellipse $E_\rho = \{z: z = \frac{1}{2}(\rho e^{i\theta} + \rho^{-1}e^{-i\theta}\}$.
Using the same error estimates given in \cite{Low2017}, we may bound
\begin{align}
M \leq \max_{z \in E_\rho}\exp \Real \left(-(k(z - \eta))^2\right)
\end{align}
Over the range $0 \leq \theta < 2\pi$, the exponent $\alpha = \Real (k(z + \eta))^2$ achieves its maximum value of 
\begin{align}
    \alpha_0 = \frac{k^2}{4\rho^2(1 + \rho^4)}(\rho^2 -1)^2(1 - 2\eta^2\rho^2 + \rho^4)
\end{align}

If we choose $\rho = 1 + a$ and expand around small $a$, then we find that $\alpha_0 = 2k^2(1 - \eta^2)a^2 + O(a^3)$. Requiring that $\alpha_0 = O(1)$ fixes the magnitude of $a$ to
\begin{align}
    a = O\left(\frac{1}{\sqrt{k^2 (1 - \eta^2)}} \right) = O\left(\frac{\delta}{\sqrt{1 - \eta^2}\log^{1/2} 1/\epsilon} \right)
\end{align}
after substituting in $k = O(\delta^{-1}\log^{1/2}(1/\epsilon))$. Note that so long as $|\eta| < 1 - \delta$, which is always true, $a$ is indeed asymptotically small and scales at least as strongly as $O(\sqrt{\delta})$. We are interested in the case where $|\eta|$ is near 1, so we may set $\sqrt{1 - \eta^2} = O(\sqrt{1 + \eta}\sqrt{1 - \eta}) = O(\sqrt{1-|\eta|})$.
Thus, the approximation error from (\ref{eq:1}) is
\begin{align}
    \epsilon = O(a^{-1}e^{-na}) \implies n &= O\left(\frac{\sqrt{1 - |\eta|}}{\delta} \log\frac{\sqrt{1 - |\eta|}}{\delta}\log^{3/2}\frac{1}{\epsilon} \right) \\
    &= \widetilde O\left(\frac{\sqrt{1 - |\eta|}}{\delta}\right)
\end{align} as desired.

\end{proof}

The Lemma shows that $p(x)$, the polynomial that we will block-encode to implement ground state projection, is exactly the order-$n$ Chebyshev approximation to the (regulated, shifted) step function $\Theta_k$; an essential difference from prior implementations is that when the shift is large, the degree $n$ is asymptotically smaller for fixed $\epsilon$ than at zero shift.

To make use of this result, we will always put the location of the shift between the ground state energy $\mu$ and the first excited state $\mu + \delta$, so that $\eta = \mu + \delta/2$. In particular, if $\mu = -1 + O(\delta)$, then the degree is $n = \widetilde O(\frac{1}{\sqrt{\delta}})$. Essentially, if the location $\mu$ of the unit step is close to $-1$ (relative to the width $\delta$) then the dependence of $n$ on $\delta$ can be reduced by up to a square-root factor relative to the unshifted case where $n = \widetilde O(\frac{1}{\delta})$.

\begin{figure}[ht]
\centering
\includegraphics[width=0.98\textwidth]{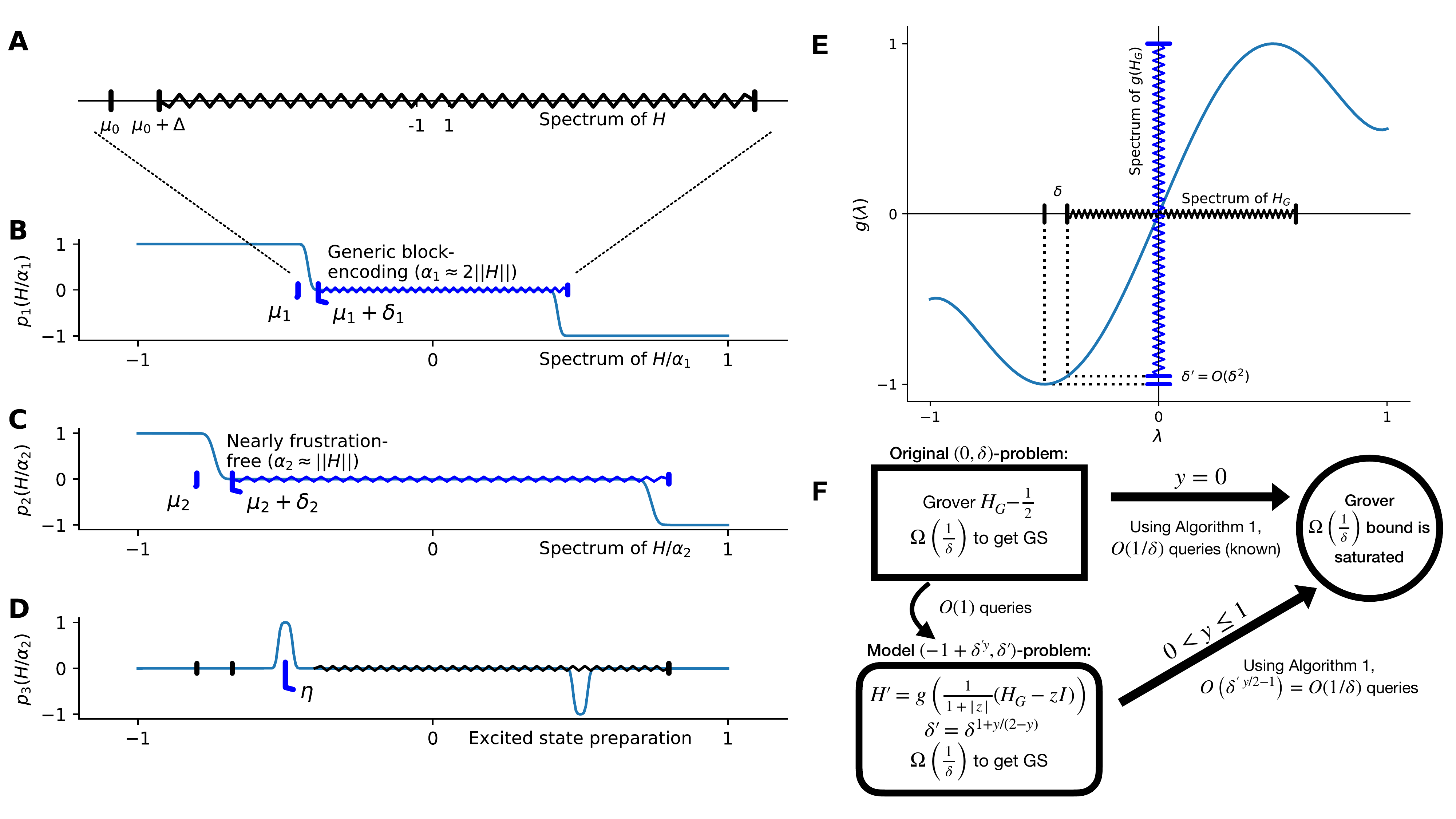}

  \caption{In A, the spectrum of $H$, the Hamiltonian whose ground state $\ket{\mu_0}$ we would like to prepare. To block-encode, $H$ must be normalized to fit its spectrum in $[-1,1]$ as shown in B and C, along with the polynomial functions $p_1(H/\alpha_1)$ and $p_2(H/\alpha_2)$ that increase the effective gap to $\Omega(1)$. In B, the normalization $\alpha_1$ is such that $\mu_1 = \mu_0/\alpha_1$ is far from $-1$ compared to the normalized gap $\delta_1 = \Delta/\alpha_1$, while in C another normalization $\alpha_2$ gives us $\mu_2 < -1 + \sqrt{\delta_2}$. Because Chebyshev polynomials oscillate faster near the endpoints of the unit interval, $p_2$ requires a lower degree than does $p_1$. D shows the analogous process for preparing excited states. In E, we show the process of generating $g(\tfrac{1}{1 + |z|}(H_G - zI))$, with $z$ defined below, whose spectrum is shown on the vertical axis. This defines a model $(\mu, \delta)$-problem that we use to show optimality. Although its ground state is close to -1, note that the new gap is asymptotically smaller. Diagram F shows the logical flow of the optimality proof for Proposition 2. The original Grover Hamiltonian with gap $\delta = \frac{2}{\sqrt{N}}$ must take $\Omega(1/\delta)$ time to solve.  Using this Hamiltonian we show optimality of our algorithm for all $0\leq y \leq 1$. Optimality of $y=0$ follows directly from the top arrow as the algorithm solves the $y=0$ problem optimally (this optimality was already known from previous works  \cite{Lin2020}).  For all other $y$, we map $H_G$ to $H'$, with gap $\delta^{1+y/(2-y)}$, which when solved by our algorithm finds the marked element of Grover's in time $1/\delta$, demonstrating optimality for all $y$.  }
  \label{fig:grover-graph}

\end{figure}

\section{Projecting to the Ground State}

By using the Chebyshev approximation of $f(x)$ from Lemma 1 as our projector, we can apply this speedup to the ground state preparation problem, identifying $\mu$ with the (normalized) ground state energy and $\delta$ with the spectral gap. Because QSP works best with polynomials of definite parity, we need to do the projection in two steps. First, using the odd component of $p(x)$, we generate a new ``effective gap'' in the spectrum of size $\delta' = 1 + O(\epsilon)$; this has query complexity $\widetilde O(\frac{\sqrt {1 + \mu + \delta/2}}{\delta})$.
Second, we use the method of \cite{Lin2020} to project into the ground state; this has complexity $\widetilde O(\frac{1}{\delta'}) = O(1)$. We compose these two quantum circuits and thus multiply their complexities, for a final query complexity that has a dependence on $\mu$ and $\delta$ of $\widetilde O(\frac{\sqrt {1 + \mu + \delta/2}}{\delta})$. This algorithm furnishes a proof of Proposition 1 when we have $\mu = -1 + \delta^y$ for some $0 \leq y \leq 1$.

\begin{proof}{Proof of Proposition 1.}
Fix $H_0$ and its $(\alpha, m, 0)$ block-encoding $U_1$, and denote  $p_\text{odd} = \frac{1}{2}(p(x) - p(-x))$ for the order-$n$ polynomial $p$ defined in Lemma 1. Given $U_1$, there is a $(1, m+1, 0)$ block-encoding $U_2$ of $p_\text{odd}(H)$ that uses $n$ queries to $U_1$ and $U_1^\dagger$ as well as $O((m+1)n)$ other primitive quantum gates \cite{Gilyen2019}. By construction, $p(\mu) \leq -1 + O(\epsilon)$ and $|p(\lambda)| = O(\epsilon)$ for any eigenvalue with $\lambda > \mu$, so the spectral gap of $p_\text{odd}(H)$ is now $\Omega(1)$. Thus, we may use the algorithm of \cite{Lin2020} to project into the ground state $\ket{\mu}$ using $O(\frac{1}{\gamma} \log \frac{1}{\epsilon})$ queries to $U_2$, or 
\begin{align}
    O\left(\frac{n}{\gamma} \log \frac{1}{\epsilon}\right) 
    = O\left(\frac{1}{\gamma}\delta^{y/2 - 1} \log^{5/2} \frac{1}{\epsilon}\log\frac{1}{\delta}\right) 
    = \widetilde O\left( \frac{1}{\gamma}\delta^{y/2 - 1} \right)
\end{align}
queries to $U_1$ and $U_1^\dagger$, as desired.
\end{proof}

\section{Estimating an Unknown Ground State Energy}
The results of the previous section assume that the ground state energy $\mu$ is known exactly; we now relax this condition and give the cost for when the exact ground state energy, and more generally the frustration-free parameter $y$, is unknown. In general, this will add additional factors logarithmic in $\gamma$ and $\delta$, and will modify the frustration-free exponent $y$.  We still require that $\gamma$ is known, although it can be assumed that only a lower-bound on $\gamma$ is known instead.

More precisely, we have the following lemma and its corollary:

\begin{lem}
    Assume that the ground state energy $\mu$ is unknown, but also that $y'$ is a known lower bound for $y$ and that $\gamma$ is a known lower bound for $|\braket{\phi|\mu_0}|$. Then the query complexity from Proposition 1 is modified to
    \[\widetilde O\left(\frac{\delta^{y'/2 - 1}}{\gamma}\right)\]
\end{lem}
\begin{cor}
    If $y$ is \textit{a priori} unknown (but $\gamma$ is known), then the query complexity from Proposition 1 is modified to
    \[\widetilde O\left(\frac{\delta^{(y - \nu)/2 - 1}}{\gamma}\right)\]
    where $\nu$ is a constant that can be made arbitrarily small with $\poly(1/\nu)$ additional overhead.
\end{cor}
Additionally, in both cases, the query complexity gains a factor of $\log \frac{1}{\delta}\,\log\frac{1}{\gamma}$.

The procedure to estimate the ground state energy is exactly that described in Section 4 of \cite{Lin2020}, with the appropriate modifications of the cost dependence on $\delta$ when operating in the nearly-frustration-free regime, augmented with monotonically increasing search for $y$ for Corollary 1. 

Informally, the idea is that the true ground state projection circuit succeeds with probability $\gamma$. Moreover, if the energy cutoff $\mu'$ of the projection circuit is set too low (meaning, $\mu' < \mu$ where $\mu$ is the true normalized ground state energy) then the circuit essentially always fails, as in that case $\Theta(H - \mu')$ is $\epsilon$-close in norm to the zero projector. We use this phenomenon to accurately estimate $\mu$ by running the circuit with different energy cutoffs, each time measuring the success probability $q$. If $q \geq \gamma$ then the current cutoff is in fact greater than $\mu$ and we thus must decrease it, and if $q \approx 0$ then we must increase it.

We should recall that success and failure in this context are referring back to qubitization/quantum signal processing/QSVT, which is the framework by which we construct the quantum circuit for our (nonunitary) $f(H)$. Because of the non-unitarity, the circuit does not always “succeed” in that the result of the circuit in the main register is conditional upon an ancilla flag register \cite{Low2019}. If the flag is measured $\ket{0}$ then the main register has been acted upon by the operator $f(H)$, in which case the circuit has “succeeded”; if it is any other state, the main register has been acted upon by a different operator and the circuit has “failed” in that the register does not contain the desired result.

We do this binary search for the true ground state energy of the nearly-frustration-free energy interval $[-1, -1 + \delta^{y'}]$, which must terminate when the precision becomes of order $\delta$. Estimating the success probability requires $O(1/\gamma)$ applications of the projection circuit,  binary search of the interval requires $O((1-y')\log \frac 1 \delta)$ applications, and each projection circuit has query complexity $O(\delta^{y'/2 - 1} \log^{3/2} \frac{1}{\epsilon})$. Multiplying these together yields the total final query complexity to estimate the ground state energy, which is identical to the ground state preparation cost (with frustration-free exponent $y'$) up to logarithmic factors.

The above discussion is actually incomplete because the described estimation of $q$ is itself probabilistic; one must perform additional repetitions to depending on the desired success probability. This adds a final logarithmic factor which appears in \cite{Lin2020} Theorem 8.

\begin{proof}{Proof of Lemma 2. }
We must search for the true ground state energy in the nearly-frustration-free interval $[-1, -1 + \delta^{y'}]$, where $y' \leq y$ is known. This may be accomplished with the binary search method described in \cite{Lin2020} Theorem 8 using the projection circuits of Proposition 1. As in that Theorem, binary search incurs an additional multiplicative overhead of $O((1-y')\log \frac{1}{\delta}\,\log\frac{1}{\gamma}\,\log\frac{\log 1/\delta}{\vartheta})$
where the success probability is $1 - \vartheta$, and the precision $\epsilon$ is taken to be of order $\delta$. The query complexity to find the ground state energy is therefore
\[O\left((1-y')\frac{\delta^{y'/2 - 1}}{\gamma} \log^{7/2}\frac{1}{\delta}\log\frac{1}{\gamma}\log\frac{\log 1/\delta}{\vartheta}\right)\]

Once the true ground state energy is known, the algorithm from Proposition 1 is subsequently employed to prepare the ground state. The final query complexity is 
\begin{align}O\left(\frac{\delta^{y'/2 - 1}}{\gamma}\left((1-y')\log^{5/2}\frac{1}{\delta}\log\frac{1}{\gamma}\log\frac{\log 1/\delta}{\vartheta} + \log^{5/2}\frac{1}{\epsilon}\right)\log\frac{1}{\delta}\right)
\end{align}
\end{proof}

To obtain the corollary, we must also determine the correct value of $y$. Binary search is not efficient here because the total cost of the algorithm will be set by the minimum $y$ among all executed ground state projection circuits, and in ``vanilla'' binary search of the interval $[0,1]$, the first attempted value is $1/2$, thereby constraining the frustration-free exponent $y$ from the complexity statement (11) to be $\geq 1/2$ which may be much larger than is actually necessary. We avoid this problem by searching for the correct value of $y$ monotonically from below. 

\begin{proof}{Proof of Corollary 1.}
Discretize the $y$ domain $[0,1]$ into $m$ bins, where $m$ is a chosen constant, by sequentially setting $y' = j/m$ for $j = m,\dotsc,1$, and then testing those values in decreasing order as lower bounds for Lemma 2. Decomposing $y = k/m + \nu$ with $k \leq m$ and $\nu < 1/m$, the algorithm will fail when $j > k$ because the ground state energy will not be found in the energy search interval. Conversely, the algorithm will succeed on step $m - k + 1 $, with frustration-free exponent $y' = k/m = y - \nu$. The exponent precision $\nu$ can be made arbitrarily small by increasing $m$, incurring an additional overhead of $1/\nu$.
\end{proof}

\section{Optimality of the Algorithm}
The algorithm of Proposition 1 achieves the optimal dependence of its query complexity on $\delta$ and $\mu$, up to a factor of $\log \frac 1 \delta$, in the case where the starting ansatz is good i.e. $\gamma = \Omega(1)$. This lower bound is  formalized by Proposition 2, and we prove it by adapting the optimality argument presented in \cite{Lin2020} and thus reducing unstructured search to the ground state preparation problem. 
We will focus on a ``Grover Hamiltonian'' $H_G$ defined as follows: given a Hilbert space of dimension $N$, define the usual marker oracle $U_t$ and Grover diffusion operator $D$ as
\begin{align}
    U_t = I - 2\ket{t}\bra{t} \quad\quad D = I - 2\ket{u}\bra{u} 
    \quad\quad \ket{u} = \frac{1}{\sqrt{N}}\sum_s \ket{s}
\end{align}

The Hamiltonian $H_G = \frac{1}{2}D + \frac{1}{2}U_t$ encodes in its ground state $\ket{\mu} \propto \ket{u} + \ket{t}$ the solution $\ket{t}$ of the corresponding unstructured search problem; the spectral gap is $\delta_G = \frac{2}{\sqrt{N}}$ and the ground state energy is $\mu_G = -\delta_G/2$. Thus, for any chosen (small) $\delta > 0$, one can choose a large enough $N > (2/\delta)^2$ such that $H_G$ has gap $\delta_G < \delta$; this fixes the ground state energy as well. We take the starting ansatz to be $\ket{u}$, which gives us $\gamma = \braket{u|\mu} \approx \frac{1}{\sqrt{2}} + O(1/\sqrt{N}) = \Omega(1)$.

Our algorithm from Proposition 1 works with ground state energies that are near -1 and, hence, a ground state energy very close to 0 doesn't achieve any speedup relative to the generic case.  One first naive attempt to prove this lower bound is to simply shift the Hamiltonian so the ground state is closer to -1.  This will improve the bound over the generic case but doesn't saturate the upper bound of Proposition 1. Instead, we perform another transformation of $H_G$ such that the algorithm of Proposition 1 on this transformed Hamiltonian solves Grover's in time $O(\sqrt{N})$; any algorithm faster than one in Proposition 1 then would violate the known lower bound of Grover's. A visual representation of this transformation is shows in Figure 1(E), and the logical flow of the proof is represented in Figure 1(F), along with corresponding discussions in the figure caption. In the rest of this section, we describe the optimality result in detail.

It is useful to measure the ground state energy relative to the endpoint $-1$, so we define the notation $\eta = O_{-1}(\mu)$ to mean that the number $\eta$ is asymptotically closer to $-1$ than is $\mu$, or more precisely that $\eta + 1 < c(\mu + 1)$ in the $\delta \to 0$ limit for constant $c$. In this section, we refer to preparing the ground state of a (normalized) Hamiltonian $H$ with ground state energy $O_{-1}(\mu)$ and spectral gap $O(\delta)$ (we will ignore small constant factors) as a $(\mu, \delta)$-problem. 

The lower bound of ref.\cite{Lin2020} for generic Hamiltonians can be represented in the following way. For any $\delta$ there is a corresponding $(0, \delta)$-problem that encodes the solution to the unstructured search, which has a query complexity lower bound of $\Omega(1/\delta)$, thus proving that the generic $(0, \delta)$-problem is $\Omega(1/\delta)$.  We need to formulate a lower bound for the general $(\mu, \delta)$-problem when $\mu$ is asymptotically close to $-1$.

We can now see that naive shifting of $H$ by a constant does not give a tight lower bound for the query complexity in this situation. With $\delta > 0$ (with $\delta \ll 1$) and $y$ such that $0 \leq y < 1$ both fixed, there is a Grover Hamiltonian $H_G$ with gap $\delta^{1-y}$, and we can shift it to  $H_\textrm{shifted}=H_G - zI$ for any real number $z$. This shift affects both $\mu$ and $\delta$ because it modifies the block-encoding normalization $\alpha$. Requiring that the $H_\textrm{shifted}$ has gap $\delta$ fixes $z$ and results in a $(-1 + \delta^y, \delta)$-problem.  Since the ground state of $H_\textrm{shifted}$ can't be prepared in time faster than $\Omega(1/\delta^{1-y})$, the generic  $(-1 + \delta^y, \delta)$-problem has a lower bound of $\Omega(\delta^{y - 1})$ queries when $y < 1$.

On the other hand, the algorithm of Proposition 1 when $\mu = -1 + \delta^y$ solves any $(-1 + \delta^y, \delta)$-problem with query complexity $O(\delta^{y/2 - 1} \polylog \tfrac{1}{\delta})$, which can be quite a bit larger than $\delta^{y - 1}$ in the $\delta \to 0$ limit that we are considering. We will need an additional step to raise the Grover-derived lower bound and prove that the $y/2 - 1$ exponent of our algorithm is indeed optimal.

\begin{proof}{Proof of Proposition 2.} The extra ingredient required is another polynomial transformation of the input $H_G$, denoted $g(x)$. We'll use $g(x)$ to map $\mu$ to approximately $-1$, and we'll see that this mapping can be made to shrink the gap $\delta$ less than the simple shift.
This polynomial must have two properties: it must have minimum value $\min_x g(x) = -1$ on the interval $x \in [-1, 1]$, and it must be one-to-one in a sufficiently large region around the ground state. Specifically, we will use the explicit polynomial $g(x) = ax + bx^3 + cx^5$ with $a = 19/6, b = -16/3, c = 8/3$, which achieves its minimum of $-1$ at $x = -\frac 1 2$. Now, we create our model $(\mu', \delta')$-problem in two steps: first, we introduce the $(-\tfrac{1}{2} + x_0, \delta)$-problem $H$, which has the same spectrum as $H_G$ but is shifted by a constant and scaled such that its ground state is exactly $-\frac{1}{2} + x_0$ for some chosen $x_0 > 0$, with $|x_0| < 1/2$:
\begin{align}
    H = \frac{1}{1 + |z|}\left(H_G - z I\right) \quad \quad z = \frac{1 - 2x_0 - \delta}{1 + 2x_0}
\end{align}
Using the linear-combination-of-unitaries approach \cite{Low2019}, $H$ can be block-encoded with $\alpha = 1$. Then, we map this $H$ through the polynomial $g$, and $g(H)$ will define a particular $(-1 + \delta'^y, \delta')$-problem with an associated lower bound in terms of $\delta'$.

We choose the offset as $x_0 = \delta^\nu$ for some $0 < \nu \leq 1$; note that $x_0$ is small in the limit $\delta\to 0$ that we are concerned with. (This shift does not asymptotically affect the gap size.) Since $g'(-\tfrac{1}{2}) = 0$ and $g''(-\tfrac{1}{2}) = O(1)$, we have
\begin{align}
    g(\mu) &= g(\tfrac{1}{2} + x_0) \simeq -1 + x_0^2 \\
    g(\mu + \delta) &\simeq -1 + x_0^2 + 2x_0\delta + \delta^2
\end{align}
Since the new gap of $g(H)$ is $g(\mu + \delta) - g(\mu)$, we have used $g$ to map the $(0, \delta)$-problem to a $(-1 + x_0^2, 2x_0\delta + \delta^2)$-problem. Since $x_0 = \delta^\nu$, to leading order in $\delta$, we've reduced the Grover problem to a $(-1 + \delta^{2\nu}, \delta^{\nu + 1})$-problem, or, rescaling $\delta$ to $\delta' = \delta^{\nu + 1}$, a $(-1 + \delta'^{\frac{2\nu}{\nu + 1}}, \delta') = (-1 + \delta'^{y}, \delta')$-problem for any $0 < y \leq 1$. 

Because $g(H)$ and $H$ have the same ground state, and because querying $g(H)$ once requires $O(1)$ queries of $H$, the original $\Omega(1/\delta)$ bound for preparing the ground state of $H$ also applies to $g(H)$. This implies that the solution to the generic $(-1 + \delta'^{y}, \delta')$-problem for any $y \leq 1$ has query complexity bounded by
\begin{align}
    \Omega(\tfrac{1}{\delta}) = \Omega\left(\delta'^{\tfrac{-1}{1 + \nu}}\right) = \Omega\left(\delta'^{\tfrac{y}{2} - 1}\right)
\end{align}
as desired.

\end{proof}

\section{Preparing Excited States with Projectors}
A straightforward generalization of the projection polynomial $p(x)$ from Section 3 allows the preparation of excited states. This extends the work of \cite{Lin2020Filtering} in that it provides a speedup for preparing low-lying excited states similar to that obtained above for nearly-frustration-free ground states.  More precisely, given a normalized excited-state energy $\eta$ (with, as before, $|\eta| \leq 1$), and a radius $\delta_1$ which separates the excitation from both the ground state and other excitations, we can form a new polynomial $p_\eta$ such that $p_\eta(H)$ selects the subspace of eigenstates with energy $\lambda \in (\eta - \delta_1, \eta + \delta_1)$. Moreover, the degree of $p_\eta$ scales exactly the same way as that of the ground state projector $p$, as $n = \widetilde O\left(\frac{\sqrt{1 + \eta}}{\delta_1}\right)$.  

\begin{cor} If a Hamiltonian $H$, with block-encoding $U$, has a unique first excited state $\ket{\eta}$ separated by an energy gap of at least $\delta_1$ from the rest of the spectrum, it can be prepared to error $\epsilon$ with $\widetilde O\left(\frac{\sqrt{1 + \eta}}{\gamma\delta_1}\right)$ queries to $U$ and $U^\dagger$, where $\gamma = |\braket{\phi_0|\eta}|$ is the overlap with the starting ansatz $\ket{\phi_0}$.

\begin{proof} As in Lemma 1, we use shifted error functions to accomplish this. Note that the function $f_1(x) = \tfrac{1}{2}(\erf k(x - \eta - \delta_1/2) - \erf k(x - \eta + \delta_1/2))$, with $k = O(\frac{1}{\delta_1}\log^{1/2}\frac{1}{\epsilon})$, satisfies $f_1(x) \leq -1 + O(\epsilon)$  when $|x - \eta| \leq \delta_1/2$ and $|f_1(x)| \leq O(\epsilon)$  when $|x - \eta| \geq \delta_1$. By the estimates of Lemma 1, $f_1(x)$ can be approximated to error $\epsilon$ with a degree-$n_1$ polynomial $p_1$, where

\begin{align}
    n_1 = O\left(\frac{\sqrt{1 + \eta}}{\delta_1}\log \frac{\sqrt{1 + \eta}}{\delta_1}\log^{3/2}\frac 1 \epsilon\right)
\end{align}
Again form the antisymmetric $p_{1,\text{odd}} =(p_1(x) - p_1(-x))$; as before, the spectral gap of $p_{1,\text{odd}}(H)$ is now $\Omega(1)$, but the ground state is now $\ket{\eta}$ with eigenvalue $\lambda_\eta = -1 + O(\epsilon)$. We can then use the algorithm of \cite{Lin2020} to project to $\ket{\eta}$ with the same query complexity as in Proposition 1, or
\begin{align}
    O\left(\frac{\sqrt{1 + \eta}}{\delta_1\gamma} \log^{5/2} \frac{1}{\epsilon}\log\frac{1}{\delta}\right) 
    = \widetilde O\left( \frac{\sqrt{1 + \eta}}{\delta_1\gamma} \right)
\end{align}
queries to $U$ and $U^\dagger$, as desired.
\end{proof}
\end{cor}

Although we can form $p_\eta$ for any $|\eta| \leq 1$, as with ground state preparation an additional speedup is possible if $\eta$ is close to the endpoints $-1$ or $1$. This could be the case for the low-lying excited states of a nearly-frustration-free $H_0$, for the following reason: Suppose the Hamiltonian $H$ has a unique first excited state $\ket{\eta}$. Then, if the ground state energy is $\mu = -1 + \delta^y$ and $y \leq 1$, then we have $\eta = \mu + \delta = -1 + \delta^y + \delta = O_{-1}(-1 + \delta^y) = O_{-1}(\mu)$. Because $\delta = O(\delta^y)$, this claimed equality follows for these shifted values relative to -1. This implies that the query complexity of this first-excited-state preparation problem has gap dependence $\widetilde O(\delta^{y/2}/\delta_1)$, where $\delta_1 = \min(\delta, \delta')$ and $\delta'$ separates $\ket{\eta}$ from the next excited state.

\section{Example Hamiltonians with Speedup}

In this section, we consider a class of Hamiltonians for which our algorithms achieve a speedup over the standard approaches.  In particular, we discuss gapped systems which have the structure of a sub-volumetric perturbation added to a frustration-free Hamiltonian. Note that for frustration free Hamiltonians, our algorithm gives the expected $1/\sqrt{\Delta}$ dependence as one would achieve from spectral gap amplification.

Sub-volumetric perturbations show up in a number of interesting physical scenarios; we now enumerate some examples.  To begin with, many physical systems have defects (i.e. line or point defects) which intercolate amongst the bulk degrees of freedom. For example, a frustration free square lattice system of size $N \times N$ perturbed by a density of $N$ impurities (i.e. approximately one per row), or a constant number of such line-defects, can receive an efficiency gain from our algorithm if the system remains gapped, as described below. Note that there is a long history of interesting physics related to impurity problems, with relevance to algorithmic methods such as DMFT \cite{RevModPhys.68.13} and NRG \cite{Pizorn2012}. 

In addition to defect problems, there is interesting physics that often happens at interfaces \cite{RevModPhys.89.025006}.
One can consider a system with two frustration-free Hamiltonians which are connected at an interface.  While the 2-d bulk of both of these models will be frustration free, the 1-d interface will have Hamiltonian terms that generically are frustrating. Our algorithm can be more efficient at finding ground states in this case. As a final example, quantum computers are expected to eventually ``replace'' classical simulation techniques for finding ground states. A very common technique in numerical methods such as DMRG is to introduce an additional pinning field \cite{Stoudenmire2012, PhysRevX.3.031010} to promote certain phases or test robustness of phase diagrams. The introduction of this pinning field on the boundary is sub-volumetric and hence has a scaling with $\Delta$ for which our algorithm applies.  

Each of these scenarios follow from the decomposition of the full system Hamiltonian into a frustration-free piece $H_0$ and a perturbation $V$.
Here we explicitly show that such Hamiltonians, so long as they remain gapped, have the key property that $\mu < -1 + \delta^y$ for $y > 0$. In general, it is difficult to determine if any given family of Hamiltonians is gapped in some limit, but for concreteness, we consider two simple scenarios below.

First, consider a Hamiltonian defined by $N^2$ frustration-free projectors $H_0 = \sum_{j=1}^{N^2} (1-\Pi_j)$, with ground state $\ket{\psi}$ satisfying $\Pi_j \ket{\psi} = \ket{\psi} \,\forall j$, and suppose $H_0$ has spectral gap $\Delta_0 \geq 1/2$. Add to this a perturbation composed of $N$ projectors $H_1 = \sum_{k=1}^N (1 - T_k)$, where each $T_k$ selects states only of energy $> 2N$, i.e. for all states $\ket{\phi}$,
\begin{align}
   \braket{\phi|H_0|\phi} \leq 2N \implies T_k\ket{\phi} = 0\, \forall k
\end{align}
Then, we may say the following about $H = H_0 + H_1$:
\begin{enumerate}
    \item $H$ has the same ground state $\ket{\psi}$, and its gap satisfies $\Delta \geq 1/2$
    \item $H$ is nearly-frustration-free with parameter $y = \tfrac{1}{2}$
    
\end{enumerate}

\begin{proof}{Proof of 1.} To begin, see that by the definitions, both $H_0 \geq 0$ and $H_1 \geq0$ as operators. Also note that all eigenstates of $H_0$ with energy $ \leq 2N$, including $\ket{\psi}$, are also eigenstates of $H_1$ and thus of $H$. If such an eigenstate has energy $E$ with respect to $H_0$, it has energy $N$ w.r.t $H_1$ and thus $E + N$ w.r.t $H$. To see this, let $\ket{\zeta}$ be such an eigenstate of $H_0$. Then $T_k\ket{\zeta} = 0 \forall k$ by hypothesis, so $H_1\ket{\zeta} = N\ket{\zeta}$ and thus $H\ket{\zeta} = (H_0 + H_1)\ket{\zeta} = (E+N)\ket{\zeta}$.

From this, we can deduce that $\ket{\psi}$ is the ground state of $H$ with energy $N$. We know already that $\ket{\psi}$ is an eigenstate of $H$ of energy $N$, but we must show that if $\ket{\zeta}$ is another an eigenstate of $H$ with energy $E$, then $E \geq N$. Letting $\ket{\zeta}$ be as such, we may immediately assume that $E \leq 2N$ (otherwise the claim follows trivially). Supposing this, then $\ket{\zeta}$ is an eigenstate of $H_0$ with energy $E - N$. This is true because of the non-negativity of $H_0$ and $H_1$: we have that $\braket{\zeta | H_0 |\zeta} \leq \braket{\zeta | H |\zeta} \leq 2N$, implying by hypothesis that $H_1 \ket{\zeta} = N\ket{\zeta}$ so that $H_0\ket{\zeta} = (H - H_1)\ket{\zeta} = (E - N)\ket{\zeta}$. Again applying non-negativity of $H_0$ yields $E\geq N$, so $\ket{\psi}$ is indeed the ground state of $H$.

It remains to be shown that the gap is still $\geq 1/2$. First, notice what happens to the low-energy ($E \leq 2N$) and high-energy $(E > 2N)$ states of $H_0$ after adding $H_1$. By the same argument in the preceding paragraph, the low-energy states remain eigenstates of $H$ with energy shifted up by $N$. The high-energy states may become mixed, but their energy always remains $> 2N$ by non-negativity of $H_1$.

The first case is when the first excited state of $H_0$ has energy $\leq N$, or $\Delta_0 \leq N$. Then it is an eigenstate of $H$ and the gap $\Delta$ between it and the ground state is unchanged from the energy shift of $N$. Finally, non-negativity of $H_1$ implies that all other states have energy larger than $N + \Delta \leq 2N$, including the other low-energy states of $H_0$ whose energy was shifted up by $N$, and any combinations of high-energy states of $H_0$, whose energy w.r.t $H$ remains $ > 2N$. Thus, $\Delta = \Delta_0$ in this case, so $\Delta \geq 1/2$.

The ``easy'' case is when $\Delta_0 > N$. Then all low-energy excited states of $H_0$ have energy $> 2N$ w.r.t $H$, as do all high-energy states of $H_0$. Thus the gap is $\Delta > N \geq 1/2$, as desired.
\end{proof}

\begin{proof}{Proof of 2.} Consider the shifted Hamiltonian $H' = \sum_j (1 - 2\Pi_j) + \sum_k (1 - 2 T_k)$, which has the same spectrum as $H$ but shifted by $-N^2 - N$ and scaled by 2. Each term of this Hamiltonian is unitary, so $H'$ may be block-encoded with $\alpha = N^2 + N = O(N^2)$. Thus, the normalized gap is $\delta = 2\Delta/\alpha \geq 2/\alpha = O(1/N^2)$ and the normalized ground state energy is $\mu = \frac{-N^2}{N^2 + N} = \tfrac{-1}{1 + 1/N} = O_{-1}(1/N) = O_{-1}(\delta^{1/2})$, as desired. 
\end{proof}

Second, consider a frustration-free $H_0$ defined on a 2D lattice of size $N \times N$ such as $H_0 = \sum_{j = 1}^{N^2} \alpha_j U_j$, where the $U_j$ may for instance be 2-body local interactions, and suppose for simplicity that $|\alpha_j| = 1$ for all $j$. Without loss of generality, we may assume that $-1$ is an eigenvalue of each $U_j$ by absorbing any phases into the coefficients $\alpha_j$, so that the ground state is $\mu = -N^2$. Say that we add frustration-free 1-D defect of $N$ perturbations $H_1 = \sum_{k=1}^N V_k = \sum_{k=1}^N  v U'_k$, where $v$ is the magnitude of each perturbation. We also require the special property that the total Hamiltonian $H = H_0 + H_1$ has spectral gap lower-bounded by a constant $\Delta$.

Under these circumstances, we can deduce that $H$ is nearly-frustration-free.
When $v$ is held fixed, we have $\alpha = N^2 + vN = O(N^2)$, and thus $\mu_0 \leq -N^2 + N v$ or $\mu \leq -1 + \frac{v}{N} = O_{-1}(1/N)$. We further have that $\delta \geq \frac{\Delta}{N^2} = \Omega(1/N^2)$ since the spectral gap is bounded from below. This is exactly the situation described in Proposition 1 with $y = 1/2$. The algorithm of that Proposition can then prepare the ground state of $H_0 + \sum_{k=1}^N V_k$ with a query complexity dependence on $N$ of
\begin{align}
    \widetilde O\left(\delta^{y/2 - 1}\right) = \widetilde O(\delta^{-3/4}) = \widetilde O(N^{3/2})
\end{align}
which is better than $\widetilde O(N^2)$ obtained without considering that $H$ is nearly frustration free. 
This argument generalizes to $d$-dimensional systems where frustrating perturbations are added sub-volumetrically.

\acknowledgements
 We acknowledge support from the Department of Energy grant DOE DESC0020165. M.T. thanks Samuel Goldman for fruitful discussions.

\bibliographystyle{unsrtnat}
\bibliography{main}

\begin{thebibliography}{20}
\providecommand{\natexlab}[1]{#1}
\providecommand{\url}[1]{\texttt{#1}}
\expandafter\ifx\csname urlstyle\endcsname\relax
  \providecommand{\doi}[1]{doi: #1}\else
  \providecommand{\doi}{doi: \begingroup \urlstyle{rm}\Url}\fi

\bibitem[Lin and Tong(2020{\natexlab{a}})]{Lin2020}
Lin Lin and Yu~Tong.
\newblock {Near-optimal ground state preparation}.
\newblock \emph{Quantum}, 4:\penalty0 1--22, 2020{\natexlab{a}}.
\newblock \doi{10.22331/Q-2020-12-14-372}.

\bibitem[Kitaev(1995)]{Kitaev1995}
A.~Yu. Kitaev.
\newblock {Quantum measurements and the Abelian Stabilizer Problem}.
\newblock pages 1--22, 1995.
\newblock \doi{10.48550/arXiv.quant-ph/9511026}.

\bibitem[Brassard et~al.(2002)Brassard, Høyer, Mosca, and Tapp]{Brassard2002}
Gilles Brassard, Peter Høyer, Michele Mosca, and Alain Tapp.
\newblock Quantum amplitude amplification and estimation.
\newblock \emph{Quantum Computation and Information}, page 53–74, 2002.
\newblock \doi{10.48550/arXiv.quant-ph/0005055}.

\bibitem[Gily{\'{e}}n et~al.(2019)Gily{\'{e}}n, Su, Low, and Wiebe]{Gilyen2019}
Andr{\'{a}}s Gily{\'{e}}n, Yuan Su, Guang~Hao Low, and Nathan Wiebe.
\newblock {Quantum singular value transformation and beyond: Exponential
  improvements for quantum matrix arithmetics}.
\newblock \emph{Proceedings of the Annual ACM Symposium on Theory of
  Computing}, pages 193--204, 2019.
\newblock \doi{10.1145/3313276.3316366}.

\bibitem[Ge et~al.(2019)Ge, Tura, and Cirac]{Ge2019}
Yimin Ge, Jordi Tura, and J.~Ignacio Cirac.
\newblock {Faster ground state preparation and high-precision ground energy
  estimation with fewer qubits}.
\newblock \emph{Journal of Mathematical Physics}, 60\penalty0 (2):\penalty0
  1--25, 2019.
\newblock \doi{10.1063/1.5027484}.

\bibitem[Grover(1996)]{Grover96}
Lov~K. Grover.
\newblock A fast quantum mechanical algorithm for database search.
\newblock In \emph{Proceedings of the Twenty-Eighth Annual ACM Symposium on
  Theory of Computing}, STOC '96, page 212–219, New York, NY, USA, 1996.
  Association for Computing Machinery.
\newblock \doi{10.1145/237814.237866}.

\bibitem[Low and Chuang(2019)]{Low2019}
Guang~Hao Low and Isaac~L. Chuang.
\newblock {Hamiltonian Simulation by Qubitization}.
\newblock \emph{Quantum}, 3:\penalty0 163, 2019.
\newblock \doi{10.22331/q-2019-07-12-163}.

\bibitem[Bravyi et~al.(2008)Bravyi, Divincenzo, Oliveira, and
  Terhal]{Bravyi2008}
Sergey Bravyi, David~P. Divincenzo, Roberto Oliveira, and Barbara~M. Terhal.
\newblock {The complexity of stoquastic local Hamiltonian problems}.
\newblock \emph{Quantum Information and Computation}, 8\penalty0 (5):\penalty0
  0361--0385, 2008.
\newblock \doi{10.26421/qic8.5-1}.

\bibitem[Somma and Boixo(2013)]{Somma2013}
R.~D. Somma and S.~Boixo.
\newblock {Spectral gap amplification}.
\newblock \emph{SIAM Journal on Computing}, 42\penalty0 (2):\penalty0 593--610,
  2013.
\newblock \doi{10.1137/120871997}.

\bibitem[Low and Chuang(2017{\natexlab{a}})]{Low2017}
Guang~Hao Low and Isaac~L. Chuang.
\newblock {Hamiltonian simulation by uniform spectral amplification}.
\newblock \emph{arXiv}, pages 1--32, 2017{\natexlab{a}}.
\newblock \doi{10.48550/arXiv.1707.05391}.

\bibitem[Chakraborty et~al.(2019)Chakraborty, Gily{\'{e}}n, and
  Jeffery]{Chakraborty2019}
Shantanav Chakraborty, Andr{\'{a}}s Gily{\'{e}}n, and Stacey Jeffery.
\newblock {The power of block-encoded matrix powers: Improved regression
  techniques via faster Hamiltonian simulation}.
\newblock In \emph{Leibniz International Proceedings in Informatics, LIPIcs},
  volume 132, pages 1--58, 2019.
\newblock \doi{10.4230/LIPIcs.ICALP.2019.33}.

\bibitem[Haah(2019)]{Haah2019}
Jeongwan Haah.
\newblock {Product decomposition of periodic functions in quantum signal
  processing}.
\newblock \emph{Quantum}, 3\penalty0 (September):\penalty0 1--22, 2019.
\newblock \doi{10.22331/q-2019-10-07-190}.

\bibitem[Low and Chuang(2017{\natexlab{b}})]{Low2017PRL}
Guang~Hao Low and Isaac~L. Chuang.
\newblock {Optimal Hamiltonian Simulation by Quantum Signal Processing}.
\newblock \emph{Physical Review Letters}, 118\penalty0 (1):\penalty0 1--6,
  2017{\natexlab{b}}.
\newblock \doi{10.1103/PhysRevLett.118.010501}.

\bibitem[Dong et~al.(2021)Dong, Meng, Whaley, and Lin]{Dong2021}
Yulong Dong, Xiang Meng, K.~Birgitta Whaley, and Lin Lin.
\newblock {Efficient phase-factor evaluation in quantum signal processing}.
\newblock \emph{Physical Review A}, 103\penalty0 (4):\penalty0 1--26, 2021.
\newblock \doi{10.1103/PhysRevA.103.042419}.

\bibitem[Lin and Tong(2020{\natexlab{b}})]{Lin2020Filtering}
Lin Lin and Yu~Tong.
\newblock {Optimal polynomial based quantum eigenstate filtering with
  application to solving quantum linear systems}.
\newblock \emph{Quantum}, 4:\penalty0 1--35, 2020{\natexlab{b}}.
\newblock \doi{10.22331/Q-2020-11-11-361}.

\bibitem[Georges et~al.(1996)Georges, Kotliar, Krauth, and
  Rozenberg]{RevModPhys.68.13}
Antoine Georges, Gabriel Kotliar, Werner Krauth, and Marcelo~J. Rozenberg.
\newblock Dynamical mean-field theory of strongly correlated fermion systems
  and the limit of infinite dimensions.
\newblock \emph{Rev. Mod. Phys.}, 68:\penalty0 13--125, Jan 1996.
\newblock \doi{10.1103/RevModPhys.68.13}.

\bibitem[Pi{\v{z}}orn and Verstraete(2012)]{Pizorn2012}
Iztok Pi{\v{z}}orn and Frank Verstraete.
\newblock {Variational numerical renormalization group: Bridging the gap
  between NRG and density matrix renormalization group}.
\newblock \emph{Physical Review Letters}, 108\penalty0 (6):\penalty0 1--8,
  2012.
\newblock \doi{10.1103/PhysRevLett.108.067202}.

\bibitem[Hellman et~al.(2017)Hellman, Hoffmann, Tserkovnyak, Beach, Fullerton,
  Leighton, MacDonald, Ralph, Arena, D\"urr, Fischer, Grollier, Heremans,
  Jungwirth, Kimel, Koopmans, Krivorotov, May, Petford-Long, Rondinelli,
  Samarth, Schuller, Slavin, Stiles, Tchernyshyov, Thiaville, and
  Zink]{RevModPhys.89.025006}
Frances Hellman, Axel Hoffmann, Yaroslav Tserkovnyak, Geoffrey S.~D. Beach,
  Eric~E. Fullerton, Chris Leighton, Allan~H. MacDonald, Daniel~C. Ralph,
  Dario~A. Arena, Hermann~A. D\"urr, Peter Fischer, Julie Grollier, Joseph~P.
  Heremans, Tomas Jungwirth, Alexey~V. Kimel, Bert Koopmans, Ilya~N.
  Krivorotov, Steven~J. May, Amanda~K. Petford-Long, James~M. Rondinelli, Nitin
  Samarth, Ivan~K. Schuller, Andrei~N. Slavin, Mark~D. Stiles, Oleg
  Tchernyshyov, Andr\'e Thiaville, and Barry~L. Zink.
\newblock Interface-induced phenomena in magnetism.
\newblock \emph{Rev. Mod. Phys.}, 89:\penalty0 025006, Jun 2017.
\newblock \doi{10.1103/RevModPhys.89.025006}.

\bibitem[Stoudenmire and White(2012)]{Stoudenmire2012}
E.~M. Stoudenmire and Steven~R. White.
\newblock {Studying two-dimensional systems with the density matrix
  renormalization group}.
\newblock \emph{Annual Review of Condensed Matter Physics}, 3\penalty0
  (1):\penalty0 111--128, 2012.
\newblock \doi{10.1146/annurev-conmatphys-020911-125018}.

\bibitem[Assaad and Herbut(2013)]{PhysRevX.3.031010}
Fakher~F. Assaad and Igor~F. Herbut.
\newblock Pinning the order: The nature of quantum criticality in the hubbard
  model on honeycomb lattice.
\newblock \emph{Phys. Rev. X}, 3:\penalty0 031010, Aug 2013.
\newblock \doi{10.1103/PhysRevX.3.031010}.

\end{thebibliography}
\end{document}